\documentclass[american,aps,pra,reprint,superscriptaddress]{revtex4-1}
\usepackage[T1]{fontenc}
\usepackage[latin9]{inputenc}
\usepackage{color}
\usepackage{babel}
\usepackage{amsmath}
\usepackage{amsthm}
\usepackage{amssymb}
\usepackage{graphicx}
\usepackage[unicode=true,pdfusetitle,
 bookmarks=true,bookmarksnumbered=false,bookmarksopen=false,
 breaklinks=false,pdfborder={0 0 0},backref=false,colorlinks=true]
 {hyperref}

\makeatletter
\theoremstyle{plain}
\newtheorem{thm}{\protect\theoremname}
\theoremstyle{plain}
\newtheorem{prop}[thm]{\protect\propositionname}
\ifx\proof\undefined
\newenvironment{proof}[1][\protect\proofname]{\par
\normalfont\topsep6\p@\@plus6\p@\relax
\trivlist
\itemindent\parindent
\item[\hskip\labelsep\scshape #1]\ignorespaces
}{%
\endtrivlist\@endpefalse
}
\providecommand{\proofname}{Proof}
\fi
\theoremstyle{plain}
\newtheorem{lem}[thm]{\protect\lemmaname}

\usepackage{txfonts}
\usepackage{braket}
\definecolor{urlcolor}{rgb}{0,0,0.7}

\makeatother

\providecommand{\lemmaname}{Lemma}
\providecommand{\propositionname}{Proposition}
\providecommand{\theoremname}{Theorem}

\begin{document}

\title{Towards resource theory of coherence in distributed scenarios}

\author{Alexander Streltsov}

\email{streltsov.physics@gmail.com}

\affiliation{Dahlem Center for Complex Quantum Systems, Freie Universität Berlin,
D-14195 Berlin, Germany}

\affiliation{ICFO -- Institut de Ciències Fotòniques, The Barcelona Institute
of Science and Technology, 08860 Castelldefels, Spain}

\author{Swapan Rana}

\author{Manabendra Nath Bera}

\affiliation{ICFO -- Institut de Ciències Fotòniques, The Barcelona Institute
of Science and Technology, 08860 Castelldefels, Spain}

\author{Maciej Lewenstein}

\affiliation{ICFO -- Institut de Ciències Fotòniques, The Barcelona Institute
of Science and Technology, 08860 Castelldefels, Spain}

\affiliation{ICREA -- Institució Catalana de Recerca i Estudis Avançats, Lluis
Companys 23, 08010 Barcelona, Spain}
\begin{abstract}
The search for a simple description of fundamental physical processes
is an important part of quantum theory. One example for such an abstraction
can be found in the distance lab paradigm: if two separated parties
are connected via a classical channel, it is notoriously difficult
to characterize all possible operations these parties can perform.
This class of operations is widely known as local operations and classical
communication (LOCC). Surprisingly, the situation becomes comparably
simple if the more general class of separable operations is considered,
a finding which has been extensively used in quantum information theory
for many years. Here, we propose a related approach for the resource
theory of quantum coherence, where two distant parties can only perform
measurements which do not create coherence and can communicate their
outcomes via a classical channel. We call this class local incoherent
operations and classical communication (LICC). While the characterization
of this class is also difficult in general, we show that the larger
class of separable incoherent operations (SI) has a simple mathematical
form, yet still preserving the main features of LICC. We demonstrate
the relevance of our approach by applying it to three different tasks:
assisted coherence distillation, quantum teleportation, and single-shot
quantum state merging. We expect that the results obtained in this
work also transfer to other concepts of coherence which are discussed
in recent literature. The approach presented here opens new ways to
study the resource theory of coherence in distributed scenarios.
\end{abstract}
\maketitle

\section{Introduction}

The resource theory of quantum coherence is a vivid research topic,
and various approaches in this direction have been presented over
the past few years \cite{Aberg2006,Gour2008,Gour2009,Levi-2014a,Baumgratz2014,Streltsov2016b}.
The formalism proposed recently by Baumgratz \emph{et al.} \cite{Baumgratz2014}
has triggered the interest of several authors, and a variety of results
have been obtained since then. One line of research is the formulation
and interpretation of new coherence quantifiers \cite{Girolami2014,Girolami2015,Yuan2015,Yao2015,Qi2015},
in particular those arising from quantum correlations such as entanglement
\cite{Streltsov2015,Killoran2015}. The study of coherence dynamics
under noisy evolution is another promising research direction~\cite{Bromley2015,Singh2015,Mani2015,Singh2015a,Chanda2015}.
The role of coherence in biological systems \cite{Lloyd-2011a,Li-2012a,Huelga-2013a},
thermodynamics \cite{Skrzypczyk2014,Aberg2014,Lostaglio2015a,Lostaglio2015b,Korzekwa2015},
spin models \cite{Karpat2014,Cakmak2015,Malvezzi2016}, and other
related tasks in quantum theory \cite{Chen2015,Cheng2015,Hu2015,Mondal2015,Deb2015,Du2015,Bai2015,Bera2015,Liu2015}
has also been investigated. 

In the framework introduced by Baumgratz \emph{et al.} \cite{Baumgratz2014},
quantum states which are diagonal in some fixed basis $\{\ket{i}\}$
are called incoherent: these are all states of the form $\rho=\sum_{k}p_{k}\ket{k}\!\bra{k}$.
A quantum operation is called incoherent if it can be written in the
form $\Lambda(\rho)=\sum_{l}K_{l}\rho K_{l}^{\dagger}$ with incoherent
Kraus operators $K_{l}$, i.e., $K_{l}\ket{m}\sim\ket{n}$, where
$\ket{m}$ and $\ket{n}$ are elements of the incoherent basis. Significant
progress within this resource theory has been achieved by Winter and
Yang~\cite{Winter2015}. In particular, they introduced the distillable
coherence and presented a closed formula for it for all quantum states.
Similar to the distillable entanglement \cite{Plenio2007,Horodecki2009},
the distillable coherence is defined as the maximal rate for extracting
maximally coherent single-qubit states 
\begin{equation}
\ket{\Psi_{2}}=\frac{1}{\sqrt{2}}(\ket{0}+\ket{1})
\end{equation}
from a given mixed state $\rho$ via incoherent operations. Another
closely related quantity is the relative entropy of coherence, initially
defined as \cite{Baumgratz2014} 
\begin{equation}
C_{r}(\rho)=\min_{\sigma\in{\cal I}}S(\rho||\sigma),
\end{equation}
where $S(\rho||\sigma)=\mathrm{Tr}[\rho\log_{2}\rho]-\mathrm{Tr}[\rho\log_{2}\sigma]$
is the relative entropy, and the minimum is taken over the set of
incoherent states ${\cal I}$. Crucially, the relative entropy of
coherence is equal to the distillable coherence and can be evaluated
exactly~\cite{Baumgratz2014,Winter2015}: 
\begin{equation}
C_{d}(\rho)=C_{r}(\rho)=S(\Delta(\rho))-S(\rho),
\end{equation}
where $S(\rho)=-\mathrm{Tr}[\rho\log_{2}\rho]$ is the von Neumann
entropy, and $\Delta(\rho)=\sum_{k}\braket{k|\rho|k}\ket{k}\!\bra{k}$
denotes dephasing of $\rho$ in the incoherent basis. 

Recently, various alternative concepts of coherence have been presented
in the literature. We will briefly review the most important approaches
in the following, and refer to~\cite{Chitambar2016,Marvian2016,deVicente2016,Streltsov2016b}
and references therein for more details. While all these notions agree
on the definition of incoherent states as states which are diagonal
in a fixed reference basis, they differ significantly in the definition
of the corresponding free operations. A notable approach in this context
is the notion of translationally invariant operations, these are operations
which commute with unitary translations $e^{-iHt}$ for some Hamiltonian
$H$ \cite{Gour2008,Gour2009}. As was shown by Marvian \emph{et al}.
\cite{Marvian2016a}, for nondegenerate Hamiltonians translationally
invariant operations are a proper subset of incoherent operations.
Moreover, translationally invariant operations have several desirable
properties, such as a free dilation: they can be implemented by introducing
an incoherent ancilla, performing a global incoherent unitary followed
by an incoherent measurement on the ancilla, and postselection on
the outcomes \cite{Marvian2016}. As was also shown in \cite{Marvian2016},
incoherent operations introduced by Baumgratz \emph{et al}.~\cite{Baumgratz2014}
in general do not have such a free dilation. While the existence of
a free dilation is clearly appealing from the resource-theoretic perspective,
the question if every reasonable resource theory should have a free
dilation is still not fully settled.

By a similar motivation, Chitambar and Gour \cite{Chitambar2016}
introduced the concept of physical incoherent operations. These operations
have a free dilation if one allows incoherent projective measurements
on the ancilla followed by classical processing of the outcomes. Interestingly,
the resource theory obtained in this way does not have a maximally
coherent state, i.e., there is no unique state from which all other
states can be obtained via physically incoherent operations. This
is also true for genuinely incoherent \cite{Streltsov2015b} and fully
incoherent operations \cite{deVicente2016}. Genuinely incoherent
operations are defined as operations which preserve all incoherent
states, they capture the framework of coherence under additional constrains
such as energy preservation \cite{Streltsov2015b}. Moreover, all
genuinely incoherent operations are incoherent regardless of a particular
experimental realization. Fully incoherent operations is the most
general set having this property \cite{deVicente2016}. 

An alternative approach to coherence was made by Yadin \emph{et al}.
\cite{Yadin2015}, who studied quantum processes which do not use
coherence. Such a process cannot be used to detect coherence in a
quantum state, i.e., an experimenter who has access to those operations
and incoherent von Neumann measurements will not be able to decide
if a quantum state has coherence or not. These operations coincide
with strictly incoherent operations, which were introduced earlier
by Winter and Yang \cite{Winter2015}. 

As has been shown in several recent works, quantum coherence plays
an important role in various tasks which are based on the laws of
quantum mechanics. One such task is quantum state merging, which was
first introduced and studied in~\cite{Horodecki2005a,Horodecki2007}.
The interplay between entanglement and local coherence in this task
was investigated very recently in~\cite{Streltsov2016}. An important
concept in this context is the notion of local quantum-incoherent
operations and classical communication~\cite{Chitambar2015}. This
class of operations is similar to the class of local operations and
classical communication where Alice and Bob can apply local measurements
and share their outcomes via a classical channel, with the only difference
that Bob's measurements have to be incoherent \cite{Chitambar2015}.

In this paper we will consider the situation where both parties, Alice
and Bob, can perform only incoherent measurement on their parts. The
corresponding class will be called local incoherent operations and
classical communication. Moreover, we will also generalize these notions
to separable operations known from entanglement theory \cite{Rains1997,Vedral1998,Horodecki2009},
thus introducing separable incoherent operations, and separable quantum-incoherent
operations. We will study the relation of all these classes among
each other, and apply them to the task of assisted coherence distillation,
which was first introduced in \cite{Chitambar2015}. We also introduce
and discuss the task of incoherent teleportation, and study the relation
between our classes on single-shot quantum state merging. As is discussed
in the conclusions of this paper, we expect that the ideas presented
in this work will find applications beyond quantum information theory,
most prominently in quantum thermodynamics and related research areas.

\section{\label{sec:SI}Classes of incoherent operations\protect \\
in distributed scenarios}

The framework of local operations and classical communication (LOCC)
is one of the most important concepts in entanglement theory, as it
describes all transformations which two separated parties (Alice and
Bob) can perform if they apply local quantum measurements and have
access to a classical channel \cite{Horodecki2009,Chitambar2014}.
These operations are difficult to capture mathematically, since a
general LOCC operation can involve an arbitrary number of rounds of
classical communication \cite{Horodecki2009,Chitambar2014}. However,
in many relevant cases it is enough to consider the informal definition
given above. In a similar fashion, we define the class of \textbf{\emph{L}}\emph{ocal
}\textbf{\emph{I}}\emph{ncoherent operations and }\textbf{\emph{C}}\emph{lassical
}\textbf{\emph{C}}\emph{ommunication} (LICC): these are LOCC operations
with the additional constraint that the local measurements of Alice
and Bob have to be incoherent. We will also consider the case where
Alice can perform arbitrary quantum measurements, while Bob is restricted
to incoherent measurements only. The corresponding class of operations
is called \textbf{\emph{L}}\emph{ocal }\textbf{\emph{Q}}\emph{uantum-}\textbf{\emph{I}}\emph{ncoherent
operations and }\textbf{\emph{C}}\emph{lassical }\textbf{\emph{C}}\emph{ommunication}
(LQICC) \cite{Chitambar2015}.

Another important framework in entanglement theory are separable operations.
While this class is larger than LOCC~\cite{Bennett1999}, it has
a simple mathematical description, and still preserves the main features
of LOCC \cite{Horodecki2009,Chitambar2014}. Separable operations
were initially introduced in \cite{Rains1997,Vedral1998} as follows:
\begin{equation}
\Lambda_{\mathrm{S}}[\rho^{AB}]=\sum_{i}A_{i}\otimes B_{i}\rho^{AB}A_{i}^{\dagger}\otimes B_{i}^{\dagger}.\label{eq:separable}
\end{equation}
The product operators $A_{i}\otimes B_{i}$ are Kraus operators, i.e.,
they fulfill the completeness condition 
\begin{equation}
\sum_{i}A_{i}^{\dagger}A_{i}\otimes B_{i}^{\dagger}B_{i}=\openone.\label{eq:completeness}
\end{equation}
The set of all separable operations will be called S. If all the operators
$A_{i}$ and $B_{i}$ are incoherent, i.e., if they satisfy the conditions
\begin{align}
A_{i}\ket{k}^{A} & \sim\ket{l}^{A},\\
B_{i}\ket{m}^{B} & \sim\ket{n}^{B},\label{eq:incoherent}
\end{align}
we will call the total operation \textbf{\emph{S}}\emph{eparable }\textbf{\emph{I}}\emph{ncoherent}
(SI)~\footnote{This should not be confused with strictly incoherent operations~\cite{Winter2015,Yadin2015}.}.
In the more general case where only Bob's operators $B_{i}$ are incoherent
-- and thus only Eq.~(\ref{eq:incoherent}) is satisfied -- the corresponding
operation will be called \textbf{\emph{S}}\emph{eparable }\textbf{\emph{Q}}\emph{uantum-}\textbf{\emph{I}}\emph{ncoherent}
(SQI). 

Having introduced the notion of LICC, LQICC, SI, and SQI operations,
we will now study the action of these operations on the initial incoherent
state $\ket{0}^{A}\ket{0}^{B}$. It is easy to see that the set of
states created from $\ket{0}^{A}\ket{0}^{B}$ via LICC and via SI
operations is the same, and given by states of the form 
\begin{equation}
\rho_{\mathrm{i}}=\sum_{k,l}p_{kl}\ket{k}\bra{k}^{A}\otimes\ket{l}\bra{l}^{B}.
\end{equation}
States of this form are known as fully incoherent states~\cite{Bromley2015,Streltsov2015},
and the set of all such states will be denoted by ${\cal I}$. Similarly,
the set of states created from $\ket{0}^{A}\ket{0}^{B}$ via LQICC
and via SQI operations is the set of quantum-incoherent states $\mathcal{QI}$.
These are states of the form \cite{Chitambar2015}
\begin{equation}
\rho_{\mathrm{qi}}=\sum_{j}p_{j}\sigma_{j}^{A}\otimes\ket{j}\bra{j}^{B}.
\end{equation}
From the above results we immediately see that SI operations map the
set of fully incoherent states $\mathcal{I}$ onto itself, and the
same is true for LICC operations. Similarly, SQI operations and LQICC
operations map the set of fully incoherent states ${\cal I}$ onto
the larger set of quantum-incoherent states $\mathcal{QI}$. These
statements are summarized in the following equalities: 
\begin{eqnarray}
\Lambda_{\mathrm{SI}}[\mathcal{I}] & = & \Lambda_{\mathrm{LICC}}[\mathcal{I}]=\mathcal{I},\label{eq:SI-1}\\
\Lambda_{\mathrm{SQI}}[\mathcal{I}] & = & \Lambda_{\mathrm{LQICC}}[\mathcal{I}]=\mathcal{QI}.\label{eq:SQI-1}
\end{eqnarray}

In general, LICC is the weakest set of operations, and the set of
separable operations S is the most powerful set of operations considered
here. Thus, we get the following inclusions:\begin{subequations}\label{eq:inclusions}
\begin{align}
\mathrm{LICC} & \subset\mathrm{SI}\subset\mathrm{SQI}\subset\mathrm{S},\label{eq:inclusion-1}\\
\mathrm{LICC} & \subset\mathrm{LQICC}\subset\mathrm{SQI}\subset\mathrm{S},\label{eq:inclusion-2}\\
\mathrm{LICC} & \subset\mathrm{LQICC}\subset\mathrm{LOCC}\subset\mathrm{S}.
\end{align}
\end{subequations} For all of the above inclusions it is straightforward
to see the weaker form $X\subseteq Y$, where $X$ and $Y$ is the
corresponding set of operations. For most of these inclusions, $X\subset Y$
can be then proven by applying the corresponding sets of operations
to the set of fully incoherent states ${\cal I}$. As an example,
$\mathrm{SI}\subset\mathrm{SQI}$ follows from the fact that $\Lambda_{\mathrm{SI}}[{\cal I}]={\cal I}$
while $\Lambda_{\mathrm{SQI}}[{\cal I}]=\mathcal{QI}$. The same arguments
apply to all the above inclusions apart from
\begin{align}
\mathrm{LICC} & \subset\mathrm{SI},\label{eq:LiccSi}\\
\mathrm{LQICC} & \subset\mathrm{SQI},\label{eq:LqiccSqi}
\end{align}
and $\mathrm{LOCC}\subset\mathrm{S}$. As already noted above Eq.~(\ref{eq:separable}),
the inclusion $\mathrm{LOCC}\subset\mathrm{S}$ was proven by Bennett
\emph{et al.} \cite{Bennett1999}, and the remaining two can be proven
using very similar arguments. In particular, Bennett \emph{et al.}
\cite{Bennett1999} presented a separable operation which cannot be
implemented via LOCC. The corresponding product operators of this
operation have the following form (see Eq.~(4) in \cite{Bennett1999}):
\begin{equation}
A_{i}\otimes B_{i}=\ket{i}\bra{\alpha_{i}}^{A}\otimes\ket{i}\bra{\beta_{i}}^{B}.\label{eq:discrimination}
\end{equation}
The particular expressions for $\ket{\alpha_{i}}$ and $\ket{\beta_{i}}$
were given in \cite{Bennett1999} (see also the Appendix), but are
not important for the rest of our proof. However, it is important
to note that the states $\ket{i}^{A}$ and $\ket{i}^{B}$ are incoherent
states of Alice and Bob respectively. It is straightforward to see
that this separable operation is also an SI operation. Moreover, since
this operation cannot be implemented via LOCC it also cannot be implemented
via LICC. This completes the proof of Eq.~(\ref{eq:LiccSi}). The
proof of Eq.~(\ref{eq:LqiccSqi}) follows by the same reasoning.

\begin{figure}
\includegraphics[width=0.8\columnwidth]{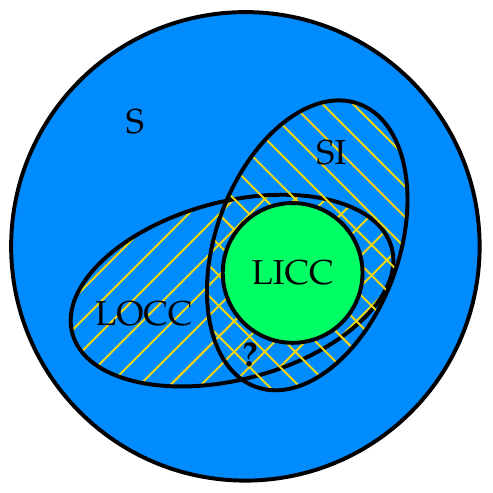}

\caption{\label{fig:hierarchy}Hierarchy of LICC, SI, LOCC, and separable operations
S. The set of LICC operations is the weakest set, and S is the most
powerful set. The crossed region displays operations which are in
LOCC and SI, but not in LICC. It remains open if such operations exist.
For simplicity, we do not display the sets LQICC and SQI. }
\end{figure}
The hierarchy of the sets LICC, SI, LOCC, and S is shown in Fig. \ref{fig:hierarchy}.
Note that the above reasoning also implies that the sets LOCC and
SI have an overlap, but one is not a subset of other. Moreover, the
figure also depicts a region of operations (crossed area) which are
simultaneously contained in LOCC and SI, but not in LICC. Such operations
would have the property that they cannot create any bipartite coherence.
On the other hand, they can be implemented via local operations and
classical communication, but for that require local coherent operations
on at least one of the parties. It remains an interesting open question
if such operations exist at all. If the answer to this question is
negative, the intersection of LOCC and SI is equal to LICC. We also
mention that similar questions arise if we consider the sets LQICC
and SQI. These sets were not included in Fig. \ref{fig:hierarchy}
for simplicity. Their inclusions are shown in Eq.~(\ref{eq:inclusions}).
In the following section we will apply the tools presented here to
the task of assisted coherence distillation, initially presented in
\cite{Chitambar2015}.

\section{Assisted coherence distillation}

\subsection{General setting}

The task of assisted coherence distillation via bipartite LQICC operations
was introduced and studied in \cite{Chitambar2015}. In this task,
Alice and Bob share many copies of a given state $\rho=\rho^{AB}$.
The aim of the process is to asymptotically distill maximally coherent
single qubit states on Bob's side. In particular, we are interested
in the maximal possible rate for this procedure.

In the following, we will extend this notion beyond LQICC operations.
For this we will consider the maximal amount of coherence that can
be distilled on Bob's side via the set of operations $X$, where $X$
is either LICC, LQICC, SI or SQI \footnote{Note that the set of LOCC operations and the set of separable operations
S cannot be included here, since they allow for local creation of
coherence on Bob's side.}. The corresponding distillable coherence on Bob's side will be denoted
by $C_{X}$ and can in general be given as follows:
\begin{equation}
C_{X}^{A|B}(\rho)=\sup\left\{ c:\lim_{n\rightarrow\infty}\left(\inf_{\Lambda\in X}\left\Vert \mathrm{Tr}_{A}\left[\Lambda[\rho^{\otimes n}]\right]-\ket{c}\bra{c}^{\otimes n}\right\Vert \right)=0\right\} ,\label{eq:Cx}
\end{equation}
where $\ket{c}=\ket{c}^{B}$ is a state on Bob's subsystem with distillable
coherence $C_{d}(\ket{c})=c$ and $||M||=\mathrm{Tr}[\sqrt{M^{\dagger}M}]$
is the trace norm. For more details regarding this definition of $C_{X}$
and for equivalent expressions we refer the reader to section \ref{sec:Remarks}.
The quantity $C_{\mathrm{LQICC}}^{A|B}$ was introduced in \cite{Chitambar2015},
where it was called \emph{distillable coherence of collaboration}. 

From Eqs.~(\ref{eq:inclusion-1}) and (\ref{eq:inclusion-2}) we
immediately see that all quantities $C_{X}$ considered here are between
$C_{\mathrm{LICC}}$ and $C_{\mathrm{SQI}}$. In the following we
will also consider the QI relative entropy which was defined in \cite{Chitambar2015}
as follows:
\begin{equation}
C_{r}^{A|B}(\rho)=\min_{\sigma\in{\cal QI}}S(\rho||\sigma).\label{eq:QIre-1}
\end{equation}
 As was also shown in \cite{Chitambar2015}, the QI relative entropy
can be written in closed form: 
\begin{equation}
C_{r}^{A|B}(\rho)=S(\Delta^{B}(\rho))-S(\rho).\label{eq:QIre-2}
\end{equation}
Note that the QI relative entropy is additive and does not increase
under SQI operations. Thus, it does not increase under any set of
operations $X$ considered here. 

It is interesting to compare the QI relative entropy to the basis-dependent
quantum discord, which was initially introduced in \cite{Ollivier2001}
and can be written as 
\begin{equation}
\delta^{A|B}(\rho)=I^{A:B}(\rho)-I^{A:B}(\Delta^{B}[\rho])
\end{equation}
with the mutual information $I^{A:B}(\rho)=S(\rho^{A})+S(\rho^{B})-S(\rho^{AB})$.
Recently, Yadin \emph{et al.} \cite{Yadin2015} have studied the role
of this quantity within the resource theory of coherence. Contrary
to the results presented in \cite{Ollivier2001}, the basis-dependent
discord vanishes on a larger set of states than the QI relative entropy.
While the latter is zero if and only if the corresponding state is
quantum-incoherent, the basis-dependent discord vanishes for all states
of the form $\rho=\sum_{i}p_{i}\rho_{i}^{A}\otimes\rho_{i}^{B}$,
where the states $\rho_{i}^{B}$ are perfectly distinguishable by
measurements in the incoherent basis \cite{Yadin2015}. This is in
particular the case if $\rho$ is quantum-incoherent or a product
state, and other examples have also been presented in~\cite{Yadin2015}.

Quite remarkably, we will see below that $C_{\mathrm{SI}}$ is equal
to $C_{\mathrm{SQI}}$ for all states $\rho$, and moreover all quantities
$C_{X}$ are bounded above by the QI relative entropy. The following
inequality summarizes these results: 
\begin{align}
C_{\mathrm{LICC}}^{A|B} & \leq C_{\mathrm{LQICC}}^{A|B}\leq C_{\mathrm{SI}}^{A|B}=C_{\mathrm{SQI}}^{A|B}\leq C_{r}^{A|B}.\label{eq:inequality}
\end{align}
The equality $C_{\mathrm{SI}}^{A|B}=C_{\mathrm{SQI}}^{A|B}$ will
be proven in the following proposition, and the bound $C_{\mathrm{SQI}}^{A|B}\leq C_{r}^{A|B}$
will be proven below in Theorem \ref{thm:bound}.
\begin{prop}
\label{prop:Csi=00003DCsqi}For an arbitrary bipartite state $\rho=\rho^{AB}$
the following equality holds:
\begin{equation}
C_{\mathrm{SI}}^{A|B}(\rho)=C_{\mathrm{SQI}}^{A|B}(\rho).
\end{equation}
\end{prop}
\begin{proof}
In the first step of the proof we will show that for any state $\rho^{AB}$
and an arbitrary SQI operation $\Lambda_{\mathrm{SQI}}$ there exists
an SI operation $\Lambda_{\mathrm{SI}}$ leading to the same reduced
state of Bob: 
\begin{equation}
\mathrm{Tr}_{A}\left[\Lambda_{\mathrm{SI}}[\rho^{AB}]\right]=\mathrm{Tr}_{A}\left[\Lambda_{\mathrm{SQI}}[\rho^{AB}]\right].\label{eq:SQI}
\end{equation}
The desired SI operation is given by 
\begin{equation}
\Lambda_{\mathrm{SI}}[\rho^{AB}]=\sum_{i}\Pi_{i}^{A}\Lambda_{\mathrm{SQI}}[\rho^{AB}]\Pi_{i}^{A},
\end{equation}
where $\Pi_{i}^{A}=\ket{i}\bra{i}^{A}$ is a complete set of orthogonal
projectors onto the incoherent basis of Alice. It is straightforward
to see that the above operation satisfies Eq.~(\ref{eq:SQI}) and
is indeed separable and incoherent. 

Now, given a state $\rho=\rho^{AB}$ with $C_{\mathrm{SQI}}^{A|B}(\rho)=c$,
for any $\varepsilon>0$ there exists an integer $n\geq1$ and an
SQI operation $\Lambda_{\mathrm{SQI}}$ acting on $n$ copies of $\rho$
such that 
\begin{equation}
\left\Vert \mathrm{Tr}_{A}\left[\Lambda_{\mathrm{SQI}}[\rho^{\otimes n}]\right]-\ket{c}\bra{c}^{\otimes n}\right\Vert \leq\varepsilon,
\end{equation}
where $\ket{c}$ is a pure state with $C_{d}(\ket{c})=c$. By using
Eq.~(\ref{eq:SQI}) it follows that for any $\varepsilon>0$ and
some integer $n\geq1$ there also exists an SI operation $\Lambda_{\mathrm{SI}}$
with the same property: 
\begin{equation}
\left\Vert \mathrm{Tr}_{A}\left[\Lambda_{\mathrm{SI}}[\rho^{\otimes n}]\right]-\ket{c}\bra{c}^{\otimes n}\right\Vert \leq\varepsilon.
\end{equation}
These arguments show that $C_{\mathrm{SI}}$ is bounded below by $C_{\mathrm{SQI}}$.
The proof of the theorem is complete by noting that $C_{\mathrm{SI}}$
is also bounded above by $C_{\mathrm{SQI}}$ since $\mathrm{SI}\subset\mathrm{SQI}$.
\end{proof}
This proposition shows that SQI operations do not provide an advantage
when compared to SI operations in the considered task: both sets of
operations lead to the same maximal performance. This result is remarkable
since the sets SI and SQI are not equal. It is now tempting to assume
that the same method can also be used to prove that $C_{\mathrm{LICC}}^{A|B}$
is equal to $C_{\mathrm{LQICC}}^{A|B}$, i.e., that quantum operations
on Alice's side do not provide any advantage for assisted coherence
distillation. Note however that the above proof does not apply to
this scenario, and thus the question remains open for general mixed
states. However, as we will see in the next section, for pure states
$C_{\mathrm{LICC}}^{A|B}$ is indeed equal to $C_{\mathrm{LQICC}}^{A|B}$.

In the following theorem we will prove that $C_{\mathrm{SQI}}$ is
bounded above by the QI relative entropy. This will complete the proof
of the inequality (\ref{eq:inequality}).
\begin{thm}
\label{thm:bound}For any bipartite state $\rho=\rho^{AB}$ holds:
\begin{equation}
C_{\mathrm{SQI}}^{A|B}(\rho)\leq C_{r}^{A|B}(\rho).
\end{equation}
\end{thm}
\begin{proof}
The proof goes similar lines of reasoning as the proof of Theorem
3 in Ref. \cite{Chitambar2015}. From the definition of $C_{\mathrm{SQI}}$
in Eq.~(\ref{eq:Cx}) it follows that for any $\varepsilon>0$ there
exists a state $\ket{\phi}$, an integer $n>1$, and an SQI protocol
$\Lambda_{\mathrm{SQI}}$ acting on $n$ copies of $\rho=\rho^{AB}$
such that 
\begin{align}
C_{\mathrm{SQI}}^{A|B}(\rho)-C_{r}(\ket{\phi}) & \leq\varepsilon,\label{eq:proof-1}\\
\left\Vert \Lambda_{\mathrm{SQI}}[\rho^{\otimes n}]-\rho_{f}^{\otimes n}\right\Vert  & \leq\varepsilon\label{eq:proof-2}
\end{align}
with the final state $\rho_{f}=\ket{0}\bra{0}^{A}\otimes\ket{\phi}\bra{\phi}^{B}$.

Eq.~(\ref{eq:proof-2}) together with the continuity of QI relative
entropy~\footnote{The QI relative entropy is continuous in the following sense~\cite{Chitambar2015}:
for any two states $\rho=\rho^{XY}$ and $\sigma=\sigma^{XY}$ with
$||\rho-\sigma||\leq1$ it holds that $|C_{r}^{X|Y}(\rho)-C_{r}^{X|Y}(\sigma)|\leq2T\log_{2}d_{XY}+2h(T)$,
where $T=||\rho-\sigma||/2$ is the trace distance, $d_{XY}$ is the
total dimension, and $h(x)=-x\log_{2}x-(1-x)\log_{2}(1-x)$ is the
binary entropy.} implies that for any $0<\varepsilon\leq1/2$ there exists an integer
$n\geq1$ and an SQI protocol $\Lambda_{\mathrm{SQI}}$ acting on
$n$ copies of $\rho$ such that
\begin{equation}
C_{r}^{A|B}(\Lambda_{\mathrm{SQI}}[\rho^{\otimes n}])\geq C_{r}^{A|B}(\rho_{f}^{\otimes n})-2n\varepsilon\log_{2}d-2h(\varepsilon),
\end{equation}
where $h(x)=-x\log_{2}x-(1-x)\log_{2}(1-x)$ is the binary entropy
and $d$ is the dimension of $AB$. Now we use the fact that the QI
relative entropy is additive \cite{Chitambar2015} and does not increase
under SQI operations. This means that for any $0<\varepsilon\leq1/2$
there exists an integer $n\geq1$ such that 
\begin{equation}
C_{r}^{A|B}(\rho)\geq C_{r}^{A|B}(\rho_{f})-2\varepsilon\log_{2}d-\frac{2}{n}h(\varepsilon).
\end{equation}
The above inequality together with the fact $C_{r}^{A|B}(\rho_{f})=C_{r}(\ket{\phi})$
and Eq.~(\ref{eq:proof-1}) implies that for any $0<\varepsilon\leq1/2$
there exists an integer $n\geq1$ such that 
\begin{equation}
C_{r}^{A|B}(\rho)\geq C_{\mathrm{SQI}}^{A|B}(\rho)-\varepsilon-2\varepsilon\log_{2}d-\frac{2}{n}h(\varepsilon).
\end{equation}
This completes the proof of the theorem.
\end{proof}
Proposition \ref{prop:Csi=00003DCsqi} and Theorem \ref{thm:bound}
in combination imply Eq.~(\ref{eq:inequality}). It remains an open
question if the inequalities in Eq.~(\ref{eq:inequality}) are strict.
As we will see in the next section, this is not the case for pure
state: in this case all quantities are equal to the von Neumann entropy
of the fully decohered state of Bob $\Delta(\rho^{B})$. 

Before we turn our attention to this question, we will first characterize
all quantum states which are useful for assisted coherence distillation
via the sets of operations $X$ considered above. Note that a quantum-incoherent
state cannot be used for extraction of coherence on Bob's side via
any set of operations $X$. On the other hand, as is shown in the
following theorem, any state which is not quantum-incoherent can be
used for extracting coherence via LICC. 
\begin{thm}
\label{thm:nonzero}A state $\rho=\rho^{AB}$ has $C_{\mathrm{LICC}}^{A|B}(\rho)>0$
if and only if it is not quantum-incoherent.\end{thm}
\begin{proof}
As $C_{\mathrm{LICC}}^{A|B}(\rho^{AB})=0$ for any quantum-incoherent
(QI) state $\rho^{AB}$, the claim follows if we show that $\rho^{AB}$
is not QI implies $C_{\mathrm{LICC}}^{A|B}(\rho^{AB})>0$. Without
loss of generality, let the non-QI state be given as 
\begin{equation}
\rho^{AB}=\sum_{i,j}\ket{e_{i}}\bra{e_{j}}^{A}\otimes N_{ij}^{B},
\end{equation}
where $\{\ket{e_{i}}^{A}\}$ is an orthonormal basis for Alice's Hilbert
space~\footnote{The states $\ket{e_{i}}$ are not necessarily incoherent.}
and $N_{ij}^{B}$ are some operators on Bob's space. By the non-QI
assumption, at least one of the $\{N_{ij}\}$ has off-diagonal elements.

We note that $N_{ii}\geq0$, so $N_{ii}\neq0\Leftrightarrow$ $\mathrm{Tr}[N_{ii}]>0$.
If an $N_{ii}$ has off-diagonal elements, then the state of Bob after
Alice's incoherent measurement with Kraus operator $K_{i}^{A}=\ket{i}\bra{e_{i}}^{A}$
is $\rho_{i}^{B}\sim N_{ii}$, with non-zero probability $\mathrm{Tr}[N_{ii}]>0$.
Hence $C_{\mathrm{LICC}}^{A|B}(\rho^{AB})>0$. 

Let us now assume that all $N_{ii}$ are diagonal. In this case --
by the non-QI assumption -- some of the operators $N_{kl}$ must have
off-diagonal elements for some $k\neq l$. If one of the operators
$N_{kl}$ (by Hermiticity of $\rho$, we can assume $k<l$ without
loss of generality) has some off-diagonal elements, then at least
one of the operators $P:=N_{kl}+N_{kl}^{\dagger}$, $Q:=i(N_{kl}-N_{kl}^{\dagger})$
will also have off-diagonal elements. Depending on whatever the case,
Alice performs an incoherent measurement containing the Kraus operator
$K_{P}:=\ket{0}\bra{e_{P}}$ or $K_{Q}:=\ket{0}\bra{e_{Q}}$, where
we define $\ket{e_{P}}:=\cos\theta\ket{e_{k}}+\sin\theta\ket{e_{l}}$,
$\ket{e_{Q}}:=\cos\theta\ket{e_{k}}+i\sin\theta\ket{e_{l}}$, the
unknown parameter $\theta$ will be determined soon. In the first
case, the post-measurement state of Bob is given by 
\begin{equation}
\rho_{\theta}^{B}\sim\cos^{2}\theta N_{kk}+\sin^{2}\theta N_{ll}+\cos\theta\sin\theta(N_{kl}+N_{kl}^{\dagger}),
\end{equation}
which, by assumption, has off-diagonal elements. Note that since $\sin^{2}\theta,\cos^{2}\theta,\sin\theta\cos\theta$
are linearly independent functions, there is always some $0<\theta<\pi/2$
for which the trace of right-hand side is non-zero, i.e., with non-zero
probability $\rho_{\theta}^{B}$ is coherent. Similarly, in the other
case, where $i(N_{kl}-N_{kl}^{\dagger})$ is assumed to have off-diagonal
elements, the post measurement state of Bob is coherent with non-zero
probability.

Thus, whenever $\rho^{AB}$ is not QI, with non-zero probability Alice
can steer Bob's state to a coherent one, which Bob can distill by
using the methods presented by Winter and Yang~\cite{Winter2015},
so $C_{\mathrm{LICC}}^{A|B}(\rho^{AB})>0$.
\end{proof}
Since LICC is the weakest set of operations considered here, this
theorem also means that a state which is not quantum-incoherent can
be used for coherence distillation on Bob's side via any set of operations
presented above.

\subsection{Pure states}

In the following we will study the scenario where the state shared
by Alice and Bob is pure, and the aim is to distill coherence at maximal
rate on Bob's side via the sets of operations presented above. Before
we study this task, we recall the definition of coherence of assistance
given in \cite{Chitambar2015}:
\begin{equation}
C_{a}(\rho)=\max\sum_{i}p_{i}C_{r}(\ket{\psi_{i}}),
\end{equation}
where the maximum is performed over all pure state decompositions
of $\rho$. We will now prove the following lemma. 
\begin{lem}
\label{lem:pure}For any pure state $\ket{\Psi}^{AB}$ there exists
an incoherent measurement on Alice's side such that: 
\begin{equation}
\sum_{i}p_{i}C_{r}(\ket{\psi_{i}}^{B})=C_{a}(\rho^{B}),\label{eq:assistance}
\end{equation}
where $\ket{\psi_{i}}^{B}$ are Bob's post-measurement states with
corresponding probability $p_{i}$.\end{lem}
\begin{proof}
Note that any pure state can be written as 
\begin{equation}
\ket{\Psi}^{AB}=\sum_{i}\sqrt{p_{i}}\ket{e_{i}}^{A}\ket{\psi_{i}}^{B},
\end{equation}
where the states $\ket{e_{i}}^{A}$ are mutually orthogonal (but not
necessarily incoherent), and the states $\ket{\psi_{i}}^{B}$ together
with probabilities $p_{i}$ fulfill Eq.~(\ref{eq:assistance}). The
desired incoherent measurement on Alice's side now consists of the
following incoherent Kraus operators: $K_{i}^{A}=\ket{i}\bra{e_{i}}^{A}$.
It can be verified by inspection that Bob's post-measurement states
indeed fulfill Eq.~(\ref{eq:assistance}). This completes the proof
of the lemma.
\end{proof}
Lemma \ref{lem:pure} implies that for any pure state $C_{\mathrm{LICC}}^{A|B}$
is bounded below by the regularized coherence of assistance of Bob's
reduced state:
\begin{equation}
C_{a}^{\infty}(\rho^{B})\leq C_{\mathrm{LICC}}^{A|B}(\ket{\Psi}^{AB}),\label{eq:bound-1}
\end{equation}
where the regularized coherence of assistance is defined as~\cite{Chitambar2015}
$C_{a}^{\infty}(\rho)=\lim_{n\rightarrow\infty}C_{a}(\rho^{\otimes n})/n$.
To prove this statement, consider the situation where Alice and Bob
share $n\cdot m$ copies of the pure state $\ket{\Psi}=\ket{\Psi}^{AB}$.
Using Lemma~\ref{lem:pure}, it follows that in the limit of large
$m$ Alice and Bob can use $n$ copies of $\ket{\Psi}$ to extract
coherence at rate $C_{a}(\rho_{B}^{\otimes n})$ on Bob's side, and
thus 
\begin{equation}
C_{a}(\rho_{B}^{\otimes n})\leq C_{\mathrm{LICC}}^{A|B}(\ket{\Psi}^{\otimes n}).
\end{equation}
The proof of Eq.~(\ref{eq:bound-1}) is complete by dividing this
inequality by $n$ and taking the limit $n\rightarrow\infty$. Equipped
with these results we are now in position to prove the following theorem.
\begin{thm}
\label{thm:pure}For any bipartite pure state $\ket{\Psi}=\ket{\Psi}^{AB}$
the following equality holds:
\begin{align}
C_{\mathrm{LICC}}^{A|B}(\ket{\Psi}) & =C_{\mathrm{LQICC}}^{A|B}(\ket{\Psi})=C_{\mathrm{SI}}^{A|B}(\ket{\Psi})\\
 & =C_{\mathrm{SQI}}^{A|B}(\ket{\Psi})=C_{r}^{A|B}(\ket{\Psi})=S(\Delta(\rho^{B})).\nonumber 
\end{align}
\end{thm}
\begin{proof}
Combining Eqs.~(\ref{eq:inequality}) and (\ref{eq:bound-1}) we
arrive at the inequality
\begin{equation}
C_{a}^{\infty}(\rho^{B})\leq C_{X}^{A|B}(\ket{\Psi})\leq C_{r}^{A|B}(\ket{\Psi}),
\end{equation}
where $X$ is one of the sets considered above. The proof is complete
by using the following equality which holds for all pure states \cite{Chitambar2015}:
\begin{equation}
C_{r}^{A|B}(\ket{\Psi})=C_{a}^{\infty}(\rho^{B})=S(\Delta(\rho^{B})).
\end{equation}

\end{proof}
This result is surprising: it implies that the performance of the
protocol does not depend on the particular set of operations performed
by Alice and Bob. In particular, the optimal performance can already
be reaches by the weakest set of operations LICC, which restricts
both Alice and Bob to local incoherent operations and classical communication.
A better performance cannot be achieved if Alice is allowed to perform
arbitrary quantum operations on her side (LQICC), and even if Alice
and Bob have access to the most general set of operations considered
here (SQI). This statement is true whenever Alice and Bob share a
pure state.

\subsection{Maximally correlated states}

Here we will consider assisted coherence distillations for states
of the form 
\begin{equation}
\rho^{AB}=\sum_{i,j}\rho_{ij}\ket{ii}\bra{jj}^{AB}.\label{eq:mc}
\end{equation}
States of this form are known as maximally correlated states~\cite{Rains2001}.
However, note that the family of states given in Eq.~(\ref{eq:mc})
does not contain all maximally correlated states, since $\ket{i}^{A}$
and $\ket{j}^{B}$ are incoherent states of Alice and Bob respectively.
We will call these states \emph{maximally correlated in the incoherent
basis}. As we show in the following proposition, also for this family
of states the inequalities (\ref{eq:inequality}) become equalities.
\begin{prop}
For any state $\rho=\rho^{AB}$ which is maximally correlated in the
incoherent basis the following equality holds:
\begin{align}
C_{\mathrm{LICC}}^{A|B}(\rho) & =C_{\mathrm{LQICC}}^{A|B}(\rho)=C_{\mathrm{SI}}^{A|B}(\rho)=C_{\mathrm{SQI}}^{A|B}(\rho)\nonumber \\
 & =C_{r}^{A|B}(\rho)=S(\Delta^{B}(\rho))-S(\rho).
\end{align}
\end{prop}
\begin{proof}
For proving this statement it is enough to prove the equality 
\begin{equation}
C_{\mathrm{LICC}}^{A|B}(\rho)=S(\Delta^{B}(\rho))-S(\rho).
\end{equation}
For this, we will present an LICC protocol achieving the above rate.
In particular, we will show that there exist an incoherent measurement
on Alice's side such that every post-measurement state of Bob has
coherence equal to $S(\Delta^{B}(\rho))-S(\rho)$. The corresponding
incoherent Kraus operators of Alice are given by $K_{j}^{A}=\ket{j}\bra{\psi_{j}}^{A}$,
where the states $\ket{\psi_{j}}$ are mutually orthogonal, maximally
coherent, and form a complete basis~\footnote{The existence of such states follows from the discrete Fourier transform.}.
Since the states $\ket{\psi_{j}}$ are all maximally coherent, they
can be written as $\ket{\psi_{j}}=1/\sqrt{d_{A}}\sum_{k}e^{i\phi_{k}^{j}}\ket{k}$
with some phases $\phi_{k}^{j}$, and $d_{A}$ is the dimension of
$A$. The corresponding post-measurement states of Bob are then given
by 
\begin{equation}
\rho_{j}^{B}=\sum_{k,l}e^{i(\phi_{l}^{j}-\phi_{k}^{j})}\rho_{kl}\ket{k}\bra{l}^{B}.
\end{equation}
If we now introduce the incoherent unitary $U_{j}=\sum_{k}e^{i\phi_{k}^{j}}\ket{k}\bra{k}$,
we see that this unitary transforms the state $\rho_{j}^{B}$ to the
state 
\begin{equation}
U_{j}\rho_{j}^{B}U_{j}^{\dagger}=\sum_{k,l}\rho_{kl}\ket{k}\bra{l}^{B}.
\end{equation}
Since the relative entropy of coherence is invariant under incoherent
unitaries, it follows that all states $\rho_{j}^{B}$ have the same
relative entropy of coherence: 
\begin{equation}
C_{r}(\rho_{j}^{B})=C_{r}\left(U_{j}\rho_{j}^{B}U_{j}^{\dagger}\right)=C_{r}\left(\sum_{k,l}\rho_{kl}\ket{k}\bra{l}^{B}\right).
\end{equation}
It is straightforward to verify that the right-hand side of this expression
is equal to $S(\Delta^{B}(\rho))-S(\rho)$, which completes the proof
of the proposition.
\end{proof}
The above proposition can also be generalized to states of the form
\begin{equation}
\rho^{AB}=\sum_{i,j}\rho_{ij}U\ket{i}\bra{j}^{A}U^{\dagger}\otimes\ket{i}\bra{j}^{B},
\end{equation}
where the unitary $U$ acts on the subsystem of Alice. In this case,
the proposition can be proven in the same way, by applying the incoherent
Kraus operators $K_{j}^{A}=\ket{j}\bra{\psi_{j}}^{A}U^{\dagger}$
on Alice's side. 

These results show that the inequality (\ref{eq:inequality}) reduces
to equality in a large number of scenarios, including all pure states,
states which are maximally correlated in the incoherent basis, and
even all states which are obtained from the latter class by applying
local unitaries on Alice's subsystem. However, it remains open if
Eq.~(\ref{eq:inequality}) is a strict inequality for any mixed state.

\subsection{\label{sec:Remarks}Remarks on the definition of $C_{X}$}

In the following we will provide some remarks on the definition of
$C_{X}$, where the set of operations $X$ is either LICC, LQICC,
SI, or SQI. First, we note that the definition of $C_{X}$ given in
Eq.~(\ref{eq:Cx}) is equivalent to the following:
\begin{equation}
C_{X}^{A|B}(\rho)=\sup\left\{ R:\lim_{n\rightarrow\infty}\left(\inf_{\Lambda\in X}\left\Vert \mathrm{Tr}_{A}\left[\Lambda\left[\rho^{\otimes n}\right]\right]-\Psi_{2}^{\otimes\left\lfloor Rn\right\rfloor }\right\Vert \right)=0\right\} ,\label{eq:Cx-1}
\end{equation}
where$\left\lfloor x\right\rfloor $ is the largest integer below
or equal to $x$ and $\Psi_{2}=\ket{\Psi_{2}}\bra{\Psi_{2}}^{B}$
is a maximally coherent single-qubit state on Bob's side. To see that
the expressions (\ref{eq:Cx}) and (\ref{eq:Cx-1}) are indeed equivalent
it is enough to note that every set of operations $X$ includes all
incoherent operations on Bob's side, and that the theory of quantum
coherence is asymptotically reversible for pure states, i.e., a state
$\ket{\psi_{1}}$ with distillable coherence $c_{1}$ can be asymptotically
converted into any other state $\ket{\psi_{2}}$ with distillable
coherence $c_{2}$ at rate $c_{1}/c_{2}$ \cite{Winter2015}.

In the above discussion we implicitly assumed that the Hilbert space
of Alice and Bob has a fixed finite dimension, and that the incoherent
operations performed by the parties preserve their dimension. One
might wonder if the performance of any of the assisted distillation
protocols $X$ discussed above changes if this assumption is relaxed,
i.e., if Alice and Bob have access to additional local incoherent
ancillas. This amounts to considering operations on the total state
of the form $\rho^{AB}\otimes\sigma^{A'}\otimes\sigma^{B'}$, where
\begin{eqnarray}
\sigma^{A'}=\ket{0}\bra{0}^{A'} & \,\,\,\,\textrm{and}\,\,\,\, & \sigma^{B'}=\ket{0}\bra{0}^{B'}
\end{eqnarray}
are additional incoherent states of Alice and Bob respectively. As
we will see below, local incoherent ancillas cannot improve the performance
of the procedure as long as SI and SQI operations are considered.
For this, we will first prove the following lemma.
\begin{lem}
\label{lem:ancillas}For any SI operation $\tilde{\Lambda}_{\mathrm{SI}}$
acting on the state $\rho^{AB}\otimes\sigma^{A'}\otimes\sigma^{B'}$
there exists another SI operation $\Lambda_{\mathrm{SI}}$ acting
on $\rho^{AB}$ such that the resulting state of $AB$ is the same:
\begin{equation}
\Lambda_{\mathrm{SI}}[\rho^{AB}]=\mathrm{Tr}_{A'B'}\left[\tilde{\Lambda}_{\mathrm{SI}}[\rho^{AB}\otimes\sigma^{A'}\otimes\sigma^{B'}]\right].\label{eq:ancillas}
\end{equation}
\end{lem}
\begin{proof}
This can be seen explicitly, by considering the form of a general
SI operation $\tilde{\Lambda}_{\mathrm{SI}}$ acting on $\rho^{AB}\otimes\sigma^{A'}\otimes\sigma^{B'}$:
\begin{align}
\tilde{\Lambda}_{\mathrm{SI}}[\rho^{AB}\otimes\sigma^{A'}\otimes\sigma^{B'}] & =\sum_{i}\tilde{A}_{i}\otimes\tilde{B}_{i}\left(\rho^{AB}\otimes\sigma^{A'}\otimes\sigma^{B'}\right)\tilde{A}_{i}^{\dagger}\otimes\tilde{B}_{i}^{\dagger},
\end{align}
where the operators $\tilde{A}_{i}$ and $\tilde{B}_{i}$ act on $AA'$
and $BB'$ respectively. The corresponding SI operation $\Lambda_{\mathrm{SI}}$
satisfying Eq.~(\ref{eq:ancillas}) is then given by
\begin{align}
\Lambda_{\mathrm{SI}}[\rho^{AB}] & =\sum_{k,l,m}A_{klm}\otimes B_{klm}\left(\rho^{AB}\right)A_{klm}^{\dagger}\otimes B_{klm}^{\dagger}.
\end{align}
The incoherent operators $A_{klm}$ and $B_{klm}$ depend on the operators
$\tilde{A}_{i}$ and $\tilde{B}_{i}$ and have the following explicit
form: 
\begin{eqnarray}
A_{klm} & = & \mathrm{Tr}_{A'}\left[\tilde{A}_{k}\left(\openone^{A}\otimes\ket{0}\bra{l}^{A'}\right)\right],\\
B_{klm} & = & \mathrm{Tr}_{B'}\left[\tilde{B}_{k}\left(\openone^{B}\otimes\ket{0}\bra{m}^{B'}\right)\right],
\end{eqnarray}
where $\{\ket{l}^{A'}\}$ is a complete set of incoherent states on
$A'$, and $\{\ket{m}^{B'}\}$ is a complete set of incoherent states
on $B'$. Using the fact that $\tilde{A}_{k}$ and $\tilde{B}_{k}$
are incoherent, it is straightforward to verify that the operators
$A_{klm}$ and $B_{klm}$ are also incoherent. Eq.~(\ref{eq:ancillas})
can also be verified by inspection.
\end{proof}
The above lemma implies that local incoherent ancillas on Alice's
or Bob's side cannot improve the performance of the protocol if SI
operations are considered. This can be seen by contradiction, assuming
that by using a state $\rho^{AB}$ and local incoherent ancillas Bob
can extract maximally coherent single-qubit states at rate $R>C_{\mathrm{SI}}^{A|B}$.
Applying Theorem \ref{thm:bound} and noting that the QI relative
entropy does not change under attaching local incoherent ancillas,
it follows that the rate $R$ is bounded above by $C_{r}^{A|B}$,
which is again bounded above by $\log_{2}d_{B}$, where $d_{B}$ is
the dimension of Bob's subsystem:
\begin{equation}
\log_{2}d_{B}\geq C_{r}^{A|B}(\rho)\geq R>C_{\mathrm{SI}}^{A|B}(\rho).
\end{equation}
The inequality $\log_{2}d_{B}\geq R$ means that if additional incoherent
ancillas would improve the procedure at all, they are not needed at
the end of the protocol and can be discarded~\footnote{In particular, this means that Bob's local system $B$ is large enough
to hold all maximally coherent states which can be potentially distilled
via SI operations.}. These results together with Lemma \ref{lem:ancillas} imply that
the rate $R$ is also reachable without additional ancillas as long
as SI operations are considered.

Very similar arguments can also be applied to the case of SQI operations.
In this case, one can prove an equivalent statement to Lemma \ref{lem:ancillas}:
for any SQI operation acting on $\rho^{AB}\otimes\sigma^{A'}\otimes\sigma^{B'}$
there exists an SQI operation acting on $\rho^{AB}$ such that the
final state of $AB$ is the same. This implies that local incoherent
ancillas cannot improve the performance in this case as well. It remains
unclear if incoherent ancillas can provide advantage for LICC or LQICC
operations. However, if Alice and Bob share a pure state, incoherent
ancillas cannot provide any advantage also in this case due to Theorem
\ref{thm:pure}.

\section{Incoherent teleportation}

In standard quantum teleportation introduced by Bennett \emph{et al}.
\cite{Bennett1993}, Alice aims to transfer her single-qubit state
to Bob by using LOCC together with one singlet. We will now consider
the task of \emph{incoherent teleportation}, which is the same as
standard teleportation up to the fact that LOCC is replaced by LICC.
This means that Alice and Bob are allowed to apply only incoherent
operations locally, and share their outcomes via a classical channel. 

It seems that the restriction to local incoherent operations provides
a severe constraint, and it is tempting to assume that Alice and Bob
will not be able to achieve perfect teleportation in this way, at
least if they have no access to additional coherent resource states.
As we will show in the following theorem, this intuition is not correct.
\begin{thm}
Perfect incoherent teleportation of an unknown state of one qubit
is possible with one singlet and two bits of classical communication.\end{thm}
\begin{proof}
To prove this statement, recall that in the standard teleportation
protocol Alice applies a Bell measurement on her qubits $A$ and $A'$
of the total initial state 
\begin{equation}
\ket{\Psi}=\ket{\psi}^{A'}\otimes\ket{\phi^{+}}^{AB},
\end{equation}
where $\ket{\phi^{+}}^{AB}=(\ket{00}+\ket{11})/\sqrt{2}$ is a maximally
entangled state and $\ket{\psi}$ is the desired state subject to
teleportation. Alice then communicates the outcome of her measurement
to Bob. Depending on the outcome of Alice's measurement, Bob either
finds his particle $B$ in the desired state $\ket{\psi}$, or he
has to additionally apply one of the Pauli matrices $\sigma_{1}$,
$i\sigma_{2}$, or $\sigma_{3}$.

It is now crucial to note that all Pauli matrices are incoherent:
$\sigma_{i}\ket{m}\sim\ket{n}$. This means that Bob can perform his
conditional rotation in an incoherent way. We will now show that also
Alice's Bell measurement can be performed in a locally incoherent
way. For this, let $\ket{\phi_{i}}$ denote the four Bell states and
consider the Kraus operators defined as $K_{i}=\ket{00}\bra{\phi_{i}}^{AA'}$.
Note that these operators are local in Alice's lab, and moreover they
are incoherent with respect to the bipartite incoherent basis of Alice.
Finally, note that these Kraus operators lead to the same post-measurement
states of Bob as the projectors $\ket{\phi_{i}}\bra{\phi_{i}}^{AA'}$.
This completes the proof of the theorem.
\end{proof}
The above theorem shows that LICC operations are indeed powerful enough
to allow for perfect teleportation. Since LICC is the weakest set
of operations considered here, the same is true also for all the other
sets LQICC, SI, and SQI: all these sets allow for perfect teleportation
of an unknown qubit with one additional singlet. These results can
be immediately extended to any system of $n$ qubits, in which case
$n$ additional singlets are required.

\section{Superiority of SQI operations in single-shot quantum state merging}

In the discussion so far, and in particular in Eqs.~(\ref{eq:inclusions}),
we have seen that the set SQI is strictly larger than LICC, LQICC,
and SI. At this point it is important to note that a larger set of
operations is not automatically more useful for real physical applications.
Nevertheless, the results presented above indeed imply the existence
of such physical tasks where the set SQI is more useful, when compared
to the other sets individually. For completeness, we will review these
results in the following.
\begin{itemize}
\item SQI is superior to SI and LICC in the task of \emph{quantum state
preparation}. In particular, by starting from an initial state $\ket{00}^{AB}$,
SQI operations can prepare all quantum-incoherent states, while only
fully incoherent states can be prepared by SI and LICC operations,
see Eqs.~(\ref{eq:SI-1}) and (\ref{eq:SQI-1}).
\item SQI is superior to LICC and LQICC in the task of \emph{quantum state
discrimination}. In particular, there exists a set of bipartite states
which can be discriminated via SQI (and also via SI), but not via
LICC and LQICC. This was discussed in detail in Section \ref{sec:SI},
see Eq.~(\ref{eq:discrimination}).
\end{itemize}

It is now interesting to note that these two arguments are unrelated,
and each of the arguments does not automatically imply the other one.
In particular, the first argument for the superiority of SQI in comparison
to SI and LICC cannot be used to show superiority in comparison to
LQICC, since the set of states that can be prepared via SQI and LQICC
is the same. On the other hand, the second argument showing superiority
of SQI in comparison to LICC and LQICC cannot be used to show superiority
in comparison to SI, since both SQI and SI are equally well suited
for the considered task, see also Section \ref{sec:SI} for more details.
It is thus natural to ask for the existence of a quantum technological
task which shows superiority of SQI operations with respect to \emph{all}
the other sets simultaneously. 

In the following, we will present such an application, which will
be based on the well-known task of quantum state merging. The latter
task was introduced and studied in \cite{Horodecki2005a,Horodecki2007},
and extended to the framework of coherence in \cite{Streltsov2016}.
In this task, three parties, Alice, Bob, and a referee share a tripartite
state $\rho^{RAB}$. The aim of the process is to send Bob's system
to Alice~\footnote{Note that in the standard approach \cite{Horodecki2005a,Horodecki2007,Streltsov2016}
Alice sends her system to Bob. Here we consider the other direction,
i.e., Bob sends his system to Alice.} while keeping the overall state intact. In contrast to \cite{Horodecki2005a,Horodecki2007,Streltsov2016},
we do not allow Alice and Bob to share any singlets, and moreover
restrict them to the sets of operations considered in this paper,
i.e., LICC, LQICC, SI, or SQI.

We will consider merging of the following tripartite state:\emph{
\begin{equation}
\rho^{RAB}=\frac{1}{9}\sum_{i}\ket{i}\!\bra{i}^{R}\otimes\ket{\psi_{i}}\!\bra{\psi_{i}}^{AB},\label{eq:merging}
\end{equation}
}where $\ket{\psi_{i}}=\ket{\alpha_{i}}\otimes\ket{\beta_{i}}$ are
nine $3\times3$ product states chosen as in Eq.~(3) in \cite{Bennett1999}
(see also the Appendix). Moreover, Alice has access to an additional
register $A'$ of dimension $3$ in an incoherent initial state $\ket{0}^{A'}$.
Alice will use this register to store Bob's system. The total final
state is given by 
\begin{equation}
\rho_{f}^{RAA'B}=\Lambda_{X}\left[\rho^{RAB}\otimes\ket{0}\!\bra{0}^{A'}\right],\label{eq:final}
\end{equation}
where $X$ denotes one of the four sets of operations considered here.
The process is successful if $\rho_{f}^{RAA'}$ is the same as $\rho^{RAB}$
up to relabeling $B$ and $A'$. Here, we consider the single-shot
scenario, i.e., the corresponding operation is applied to one copy
of the state only.

As we will now show, this task can be performed via SQI, but not via
any of the other three sets. For proving that the task can be performed
via SQI, it is enough to recall that SQI operations can be used to
distinguish the states $\ket{\psi_{i}}$. For each outcome $i$ Alice
can then locally prepare her system $AA'$ in the state $\ket{\psi_{i}}^{AA'}$,
which again can be achieved via SQI operations~\footnote{The overall SQI Kraus operators for this procedure can be given as
$K_{ij}^{AA'B}=\ket{\alpha_{i}}\!\bra{\alpha_{i}}^{A}\otimes\ket{\beta_{i}}\!\bra{j}^{A'}\otimes\ket{0}\!\bra{\beta_{i}}^{B}$.}. In the next step we will show that this task cannot be performed
via LQICC operations. This can be seen by contradiction, assuming
that LQICC operations can perform merging in this scenario, i.e.,
that

\begin{equation}
\mathrm{Tr}_{B}\left[\Lambda_{\mathrm{LQICC}}[\rho^{RAB}\otimes\ket{0}\!\bra{0}^{A'}]\right]=\frac{1}{9}\sum_{i}\ket{i}\!\bra{i}^{R}\otimes\ket{\psi_{i}}\!\bra{\psi_{i}}^{AA'},
\end{equation}
 for some LQICC protocol $\Lambda_{\mathrm{LQICC}}$. By linearity,
it must be that 
\begin{equation}
\ket{\psi_{i}}\!\bra{\psi_{i}}^{AA'}=\mathrm{Tr}_{B}\left[\Lambda_{\mathrm{LQICC}}[\ket{\psi_{i}}\!\bra{\psi_{i}}^{AB}\otimes\ket{0}\!\bra{0}^{A'}]\right],
\end{equation}
i.e., Alice and Bob could use the protocol to transfer Bob's part
of $\ket{\psi_{i}}^{AB}$ to Alice. Since the states $\ket{\psi_{i}}^{AB}$
form an orthonormal basis, this would imply that Alice and Bob could
distinguish the states $\ket{\psi_{i}}^{AB}$ via LQICC (and thus
also via LOCC), which is a contradiction to the main result of \cite{Bennett1999}.
This also proves that the task cannot be performed via LICC operations.

It now remains to show that the task cannot be performed via SI operations.
We will prove this statement by using general properties of QI relative
entropy given in Eq.~(\ref{eq:QIre-1}) and its closed expression
in Eq.~(\ref{eq:QIre-2}). In particular, recall that the QI relative
entropy cannot increase under SI operations, which implies that 
\begin{equation}
C_{r}^{RB|AA'}\left(\rho_{f}^{RAA'B}\right)\leq C_{r}^{RB|AA'}\left(\rho^{RAB}\otimes\ket{0}\!\bra{0}^{A'}\right),
\end{equation}
where $\rho_{f}$ is the final state given in Eq.~(\ref{eq:final}).
In the next step we will use the following relations: 
\begin{align}
C_{r}^{R|AA'}\left(\rho_{f}^{RAA'}\right) & \leq C_{r}^{RB|AA'}\left(\rho_{f}^{RAA'B}\right),\\
C_{r}^{RB|A}\left(\rho^{RAB}\right) & =C_{r}^{RB|AA'}\left(\rho^{RAB}\otimes\ket{0}\!\bra{0}^{A'}\right),
\end{align}
which can be proven directly from the properties of QI relative entropy.
Combining these results we arrive at the following inequality:
\begin{equation}
C_{r}^{R|AA'}\left(\rho_{f}^{RAA'}\right)\leq C_{r}^{RB|A}\left(\rho^{RAB}\right).\label{eq:contradiction}
\end{equation}

In the final step of the proof, assume that SI operations allow to
perform the aforementioned task. We will now show that this assumption
leads to a contradiction. In particular, by our assumption the final
state $\rho_{f}^{RAA'}$ must be the same as $\rho^{RAB}$, up to
relabeling $B$ and $A'$. Thus, Eq.~(\ref{eq:contradiction}) is
equivalent to 
\begin{equation}
C_{r}^{R|AB}\left(\rho^{RAB}\right)\leq C_{r}^{RB|A}\left(\rho^{RAB}\right).
\end{equation}
By applying Eq.~(\ref{eq:QIre-2}) together with the expression for
the state $\rho^{RAB}$ in Eq.~(\ref{eq:merging}) and using the
states $\ket{\psi_{i}}$ in Eq.~(3) of \cite{Bennett1999} (see also
the Appendix) we can now evaluate both sides of this inequality: 
\begin{align}
C_{r}^{R|AB}\left(\rho^{RAB}\right) & =\frac{8}{9},\\
C_{r}^{RB|A}\left(\rho^{RAB}\right) & =\frac{4}{9},
\end{align}
which is the desired contradiction. This finishes the proof that the
task considered here can be performed with SQI operations, but not
with any other set LICC, LQICC, or SI.

Thus, we presented the first example for a quantum technological application
which can be performed via SQI operations, but cannot be performed
with any of the other sets of operations considered in this paper.

\section{Conclusions}

In this paper we studied the resource theory of coherence in distributed
scenarios. In particular, we focused on the following four classes
of operations: local incoherent operations and classical communication
(LICC), local quantum-incoherent operations and classical communication
(LQICC), separable incoherent operations (SI), and separable quantum-incoherent
operations (SQI). We showed that these classes obey inclusion relations
very similar to those between LOCC and separable operations known
from entanglement theory. 

We further studied the role of these classes for the task of assisted
coherence distillation, first introduced in \cite{Chitambar2015}.
Regardless of the particular class of operations we proved that a
bipartite state can be used for coherence extraction on Bob's side
if and only if the state is not quantum-incoherent. We also showed
that the relative entropy distance to the set of quantum-incoherent
states provides an upper bound for coherence distillation on Bob's
side, a result which again does not depend on the class of operations
under scrutiny. Remarkably, both the SI and the SQI operations lead
to the same performance in this task for all mixed states. For pure
states an even stronger result has been proven: in this case all classes
of operations considered here are equivalent for assisted coherence
distillation. We also studied the task of incoherent teleportation,
which arises from standard quantum teleportation by restricting the
parties to LICC operations only. We showed that in this situation
LICC operations do not provide any restriction: perfect teleportation
of an unknown qubit can be achieved with LICC and one additional singlet.
Finally, we compared these classes on the task of single-shot quantum
state merging. In this task, SQI operations provide an advantage with
respect to all the other classes considered here.

The tools presented here can be regarded as a first step towards a
full resource theory of coherence in distributed scenarios. Indeed,
while in the course of this work we focused on the coherence framework
of Baumgratz \emph{et al.}~\cite{Baumgratz2014}, it is important
to note that the presented ideas are significantly more general. As
an example, our tools can be directly applied to the situation where
local incoherent operations are replaced by another well justified
set, such as strictly incoherent~\cite{Winter2015,Yadin2015}, translationally
invariant~\cite{Gour2008,Gour2009,Marvian2016a,Marvian2016}, or
physical incoherent operations \cite{Chitambar2016}. These alternative
frameworks have been extensively studied in recent literature, and
each of them captures the concept of coherence in a different scenario~\cite{Streltsov2016b}.
Due to the close connection between all these frameworks it is clear
that the ideas presented in this work also carry over to these concepts.
In particular, the results presented here will in general serve as
bounds for other concepts of coherence. We also expect that these
bounds are tight in several relevant situations, e.g., for pure states.
While the investigation of these questions is beyond the scope of
this work, it has the potential to provide a unified view for all
frameworks of coherence, and ultimately put the resource theory of
coherence on equal footing with other quantum resource theories, most
prominently the theory of entanglement~\cite{Horodecki2009}.

Future research in this direction is also important in the light of
the recent progress towards understanding the role of coherence in
quantum thermodynamics~\cite{Lostaglio2015a,Lostaglio2015b}. Here,
the framework of thermal operations turned out to be very useful~\cite{Janzing2000}.
These operations arise from the first and second law of thermodynamics,
and are known to be translationally invariant~\cite{Lostaglio2015a,Lostaglio2015b}.
Because of this, the tools developed in our work can also be applied
to quantum thermodynamics. This research direction has the potential
to reveal new surprising effects, similar to well-known phenomena
such as bound entanglement~\cite{Horodecki1997,Horodecki1998} in
quantum information theory, or the work-locking phenomenon~\cite{Horodecki2013b,Skrzypczyk2014,Lostaglio2015a,Korzekwa2015}
in quantum thermodynamics.

After the appearance of this article on arXiv, quantum coherence in
multipartite systems has also been discussed by other authors~\cite{Yadin2015,Ma2016,Matera2016}.
While these works also study the role of coherence for quantum-state
manipulation, their motivation is significantly different from the
concept presented here. In particular, the framework of coherence
in multipartite systems is naturally suited for studying general nonclassical
correlations such as quantum discord, as discussed in~\cite{Ma2016,Yadin2015}.
Another important research direction pursued in~\cite{Ma2016,Matera2016}
is the role of coherence in the DQC1 protocol~\cite{Knill1998}.
This quantum protocol allows for efficient evaluation of the trace
of a unitary, provided that the unitary has an efficient description
in terms of two-qubit gates. Remarkably, this protocols does not require
entanglement, while showing an exponential speedup over the best known
classical procedure~\cite{Datta2005}. The authors of~\cite{Ma2016}
present a figure of merit for this task, which is related to the consumption
of coherence in this protocol. Interestingly, while their figure of
merit vanishes for unitaries of the form $U=e^{i\phi}\openone$, it
is unclear if a classical algorithm can evaluate the trace of this
unitary efficiently~\cite{Ma2016}. In the light of these results,
we expect that the tools presented in our work can also find applications
for the DQC1 protocol and quantum computation in general. This is
beyond of the scope of the current work, and we leave it open for
future research.

We also note that the class of LICC operations was introduced independently
by Chitambar and Hsieh \cite{Chitambar2015a}. The authors study the
tasks of asymptotic state creation and distillation of entanglement
and coherence via LICC operations, i.e., local coherence is considered
as an additional resource. They also independently obtain the results
of our Theorems \ref{thm:nonzero} and \ref{thm:pure} for LICC operations.
\begin{acknowledgments}
We thank Gerardo Adesso, Eric Chitambar, Jens Eisert, Gilad Gour,
Iman Marvian, Martin B. Plenio, and Andreas Winter for discussion.
We acknowledge financial supports from the Alexander von Humboldt-Foundation,
the John Templeton Foundation, the EU grants OSYRIS (ERC-2013-AdG
Grant No. 339106), QUIC (H2020-FETPROACT-2014 No. 641122), and SIQS
(FP7-ICT-2011-9 No. 600645), the Spanish MINECO grants (FIS2013-46768,
FIS2008-01236, and FIS2013-40627-P) with the support of FEDER funds
and \textquotedblleft Severo Ochoa\textquotedblright{} Programme (SEV-2015-0522),
and the Generalitat de Catalunya grant (2014-SGR-874 and 2014-SGR-966),
and Fundació Privada Cellex.
\end{acknowledgments}

\appendix

\section*{Appendix}

Here, we list the states $\ket{\psi_{i}}=\ket{\alpha_{i}}\otimes\ket{\beta_{i}}$
from Eq.~(3) of \cite{Bennett1999}:

\begin{align*}
\ket{\psi_{1}} & =\ket{\alpha_{1}}\otimes\ket{\beta_{1}}=\ket{1}\otimes\ket{1},\\
\ket{\psi_{2}} & =\ket{\alpha_{2}}\otimes\ket{\beta_{2}}=\ket{0}\otimes\frac{1}{\sqrt{2}}(\ket{0}+\ket{1}),\\
\ket{\psi_{3}} & =\ket{\alpha_{3}}\otimes\ket{\beta_{3}}=\ket{0}\otimes\frac{1}{\sqrt{2}}(\ket{0}-\ket{1}),\\
\ket{\psi_{4}} & =\ket{\alpha_{4}}\otimes\ket{\beta_{4}}=\ket{2}\otimes\frac{1}{\sqrt{2}}(\ket{1}+\ket{2}),\\
\ket{\psi_{5}} & =\ket{\alpha_{5}}\otimes\ket{\beta_{5}}=\ket{2}\otimes\frac{1}{\sqrt{2}}(\ket{1}-\ket{2}),\\
\ket{\psi_{6}} & =\ket{\alpha_{6}}\otimes\ket{\beta_{6}}=\frac{1}{\sqrt{2}}(\ket{1}+\ket{2})\otimes\ket{0},\\
\ket{\psi_{7}} & =\ket{\alpha_{7}}\otimes\ket{\beta_{7}}=\frac{1}{\sqrt{2}}(\ket{1}-\ket{2})\otimes\ket{0},\\
\ket{\psi_{8}} & =\ket{\alpha_{8}}\otimes\ket{\beta_{8}}=\frac{1}{\sqrt{2}}(\ket{0}+\ket{1})\otimes\ket{2},\\
\ket{\psi_{9}} & =\ket{\alpha_{9}}\otimes\ket{\beta_{9}}=\frac{1}{\sqrt{2}}(\ket{0}-\ket{1})\otimes\ket{2}.
\end{align*}

\bibliographystyle{apsrev4-1}
\bibliography{literature}

\end{document}